\documentclass[aps,pre,twocolumn,groupedaddress, showpacs]{revtex4-1}
\usepackage[dvips]{graphicx}
\usepackage{amsmath}
\usepackage{amsfonts}
\usepackage{amsthm}

\begin{document}
\renewcommand*{\figurename}{FIG.}
\newtheorem{thm}{Theorem}
\newtheorem{definition}{Definition}
\newtheorem{prop}{Proposition}
\newtheorem{cor}{Corollary}

\title{Tenement house model}

\author{Wojciech Ganczarek}
\email[]{w.ganczarek@uj.edu.pl}
\affiliation{Institute of Physics, Jagiellonian University, ul.\
Reymonta 4, 30-059 Krak\'ow, Poland\\
Institute of Mathematics, Jagiellonian University, ul.\
\-Lojasiewicza 6, 30-346 Krak\'ow, Poland
 }

\date{November 24, 2012}

\begin{abstract}
Most of the common used models of epidemic spreading allow contaminating many neighbors of a particular node in the network. They are usually analyzed by differential equations on probability vectors. We propose a model of epidemic spreading, which restricts to at most one contamination per time step and analyze it by discrete approach, working on vectors of possible states of the system. Theoretical predictions of epidemic treshold, stationary state and time needed to reach it are given and appear to be perfectly consistent with computer simulations. We also point ou that the model appears to be well suited to mimic epidemic spreading within student communes.
\end{abstract}


\maketitle

\section{Introduction}
It has been admitted, that the most appropiate models of epidemics spreading are these based on dynamical processes on particular graph models rather than those defined by fenomenological differential equations \cite{AM,NewmanSIAM}. Within this approach the nodes of a network are usually considered as individuals, who are connected with each other by vertices corresponding to social links.  Although some authors use continous time simulations (see e.g. \cite{Boguna}), the approach presented commonly (see \cite{DynProc} for a review) is based on the idea that at each discrete time step a particular node of the network can contaminate each of its neighbors with some finite probability $p$. The whole set of vertices is being divided into compartments, usually referred to susceptible (S), infected (I) and recovered (R) individuals, but the general mechanism stays more or less unchanged. There has been a broad range of methods developed in order to analyze such models. In the most basic approach people assume individuals to be identical and homogeneously mixed (homogeneous assumption, 
\cite{DynProc}). In order to take into account heterogenity of the system a kind of block approximation hax been used \cite{PSV}, treating nodes with the same degree as statistically equivalent. This is not always enough, as some real networks manifest degree correlation, mainly: the conditional probability, that two vertices of degree $k,\,k'$ are connected depends on both degrees $k,\,k'$ \cite{cor1}. The next step thus is to take into account correlation \cite{BPSV}. Finally, one can employ whole adjacency matrix describing the graph we analyze \cite{FullM,cor2,cor3}. The validity of all these approaches is still under investigation, see e.g. \cite{Localization}. Note however, that all these variations listed above work on equations describing relationships between probability vectors. In particular, for the last example, the system is being described by $p_i$ - probability, that i-th node is infected. The problem, however, is that there is not a single moment when a particular verte is - say - 0.41 
infected. A vertex can be either infected (1) or not (0). This problem has been already noticed by Petermann and De Los Rios \cite{Paulo}.

In this paper we introduce another model of epidemic spreading and analize it with complitely different approach. Let us focus on sexual transmitted diseases. For this case the assumption that a particular node is able to contaminate more then one of its neighbots during a time step seems not to be the most suitable one. Bearing this idea in mind we develope a single infection epidemic spreading model. 

This paper is organized as follows. In section II we describe proposed model of epidemics spreading with at most one infection per time step. In Section III the theoretical analysis of the model: epidemics treshold, stationary state and mixing time, is being investigated. Simulations are presented in section IV. In Section V we draw the final considerations.

\section{Model description}
\label{des}
Consider a connected, unweighted graph with $n$ vertices enumerated by indices $i=1,\,\dots,\,n$, described by transition matrix $\{P_{ij}\}$, $\sum_{j=1}^n P_{ij}=1$. The model will be of SIS kind: all the individuals are at the beginning considered as susceptible (S). After contamination they become infected (I) but they still have a chance to recover and be susceptible again.

We start thus with the all but one nodes susceptible. The one which is infected is chosen at random. The whole process consists of 4 actions which we repeat at each discrete time step.  At each time step we choose randomly, with identical probability $\frac{1}{n}$, a node - say - $i$-th one. Then we choose its neighbor according to the transition matrix $\{P_{ij}\}$, i.e. there is $P_{ij}$ chance that we point $j$-th vertex. If one of these two individuals $i,\,j$ is infected, it contaminates the second one with probability $z$. At the end of each time step we recover each infected node with probability $r$.

Note, that this method restricts not only each infected node to contaminate at most one of its neighbor. In fact we restrict all the dynamics to at most one contamination per time step. One can say it is not realistic approach. However, from the one hand we can say that this could be a basis for further generalisation. From the other: we can imagine, and even find in reality, systems that fulfill assumption described above. In large, academic cities there are often big flats situated in old tenement houses, settled by quite large amounts of students, who live with 3-4 roommates per chamber. As there is no space for privacy in this way of living, they sometimes devote one room in the flat to be a so-called {\sl sexroom}, so contamination by sexually transmitted diseases can take place at most once per time step (say: per hour). This seems to be a good example of a system which can be described by our model.

\section{Model analysis}
In order to mathematically describe the model we define $X_j(t)$ which takes the value $1$ if the node $j$ is being contaminated by one of its neighbors at the time step $t$, and $0$ otherwise. Furthermore, we denote the set of all infected nodes at the time step $t$ by $I(t)$. We will be interested in the expectation value of $X_j(t)$ with a condition that the set of infected nodes consists of some particular vertices. 

There are two independent ways of contaminating $j$-th node during one time step. Either we choose $j$-th node (with probability $\frac{1}{n}$) and then one of its infected neighbor (with probability $\sum_{k\in I(t)}P_{jk}$) or we choose $j$-th node's neighbor (with probability $\frac{1}{n}$ for each one) and then we pick $j$-th node (it happens with probability $P_{kj}$ for a particular node $k$, so $\sum_{k\in I(t)}P_{kj}$ for all of them). Summing up we obtain:
\begin{equation}
  \mathbb{E}(X_j|I(t))=\frac{z}{n}\Big(\sum\limits_{k\in I(t)}P_{jk}+\sum\limits_{k\in I(t)}P_{kj}\Big),
\label{eq:exp}
\end{equation}
where both terms are multiplied by contamination probability $z$. We are, however, interested in the behaviour of whole system, not one node only. 

Let us thus define $D(t)$ - expectation value of change of the number of infected nodes. Due to additivity of expectation value we can write:
\begin{equation}
D(t)=\mathbb{E}(|I(t+1)|-|I(t)|\,|\,I(t))=\sum\limits_{j\notin I(t)}\mathbb{E}(X_j|I(t))-r|I(t)|,
\end{equation}
where, apart from adding all $\mathbb{E}(X_j|I(t))$ terms, we substract the term responsible for healing: number of infected nodes multiplied by recovery probability $r$. Using Eq. (\ref{eq:exp}) we immediately conclude:
\begin{equation}
D(t)=\frac{z}{n}\Big(\sum\limits_{k\in I(t),\,j\notin I(t)}P_{jk}+\sum\limits_{k\in I(t),\,j\notin I(t)}P_{kj}\Big)-r|I(t)|.
\label{eq:d}  
\end{equation}
The equation above defines the dynamics of the model: by solving it one could provide the complete information about the process.
Unfortunately, in general sums $\sum_{k\in I(t),\,j\notin I(t)}P_{jk}$, $\sum_{k\in I(t),\,j\notin I(t)}P_{kj}$ are not 
precisely known as they strongly depend on the shape of the set $I(t)$. We will show, however, that we are able to derive
exact result for epidemic treshold for any graph and stationary state for some special cases.

\subsection{Epidemic treshold}
Our first aim is to find out the epidemic treshold for the process described above. We are interested in some relation of model parameters $n,\,z,\,r$ that defines a border between two situations: dropping and rising of the number of infected nodes in the beginning of the process.

We are going to analyze Eq. (\ref{eq:d}). We have there two sums that look similar, so the first idea would be to add them somehow. But in general $\sum_{k\in I(t),\,j\notin I(t)}P_{jk}\ne\sum_{k\in I(t),\,j\notin I(t)}P_{kj}$, so we cannot that easily simplify this equation (except for $\{P_{ij}\}$ - bistochastic). However, in order to find epidemic treshold, we are interested in the behaviour of the system in the vicinity of $t=0$. Let us thus take $|I(t)|=1$ then, as it happens at the very beginning of the evolution, and denote the only infected neighbor by index $l$. Then $\sum_{k\in I(t),\,j\notin I(t)}P_{kj}=\sum_{j\notin I(t)}P_{lj}=1$ (as $\{P_{ij}\}$ - stochastic) and $\sum_{k\in I(t),\,j\notin I(t)}P_{jk}=\sum_{j\notin I(t)}P_{jl}$. Now we use the fact that at the beginning the first infected node is being chosen uniformly at random. Therefore the last term should be averaged over all possibilities of choosing $l$: $\frac{1}{n}\sum_{l\in V}\sum_{j\notin I(t)}P_{jl}=\frac{1}{n}\sum_{j\notin I(t)}1=\
frac{n-1}{n}$, where we used once again the fact, that $\{P_{ij}\}$ is stochastic. Finally we write the condition $D(t)\geq0$ which indicates the epidemic outbreak:
\begin{equation}
D(t)\leq\frac{z}{n}\times (1+\frac{n-1}{n})-r\leq0.
\end{equation}
For large $n$ the term $\frac{n-1}{n^2}$ can safely be substituted by $\frac{1}{n}$. The condition for epidemic treshold for the model we proposed is thus:
\begin{equation}
\frac{z}{r}=\frac{n}{2}.
\label{eq:trsh}
\end{equation}

\subsection{Stationary state}
Let us now turn to stationary state problem. The model being analized is by definition a purely Markovian one and above the epidemic treshold we anticipate our system to stay at some non-zero stationary state, i.e. we expect that the number of infected nodes will, in long times, oscillate about a fixed value. Practically however, due to statictical flucutation in finite real or simulational system, the epidemy may die out even above the treshold. 
                                                                                                                                                                                                                                                 
The stationary fraction of infected nodes in general case (not specifying any particular shape of the graph) is not as easy reachable as the treshold calculated in the last section. What we basically have to do is to use once again all the formalism presented above and find the solution for the equation $D(t)=0$ without the constraint $|I(t)|=1$. The problem is to compute the sum $\sum_{k\in I(t),\,j\notin I(t)}P_{jk}$ - a task which is not trivial. We will thus estimate only stationary state for general case. In later subsections we give exact solutions for special cases of complete graph and uncorrelated homogenous graph.

In order to perform estimation of the stationary state, we introduce the notion of graph conductance \cite{Con}: 
\begin{definition} Conductance of a given graph $G$ described by a stochastic matrix $\{P_{ij}\}$ is:
\begin{equation}
\Phi (P)=\min_{S\subset V}\frac{\sum\limits_{j\in S,\,k\notin S}P_{jk}}{min\{|S|,|V-S|\}},  
\label{eq:con}
\end{equation}
where $V$ is the set of vertices of a graph G.
\end{definition}
This quantity measures how well-connected a given graph is. Due to the definition above we will analyze separately cases with the stationary fraction of infected nodes $i_s=\frac{|I_s|}{n}$ smaller and greater than $\frac{1}{2}$.

Consider first $i_s\geq\frac{1}{2}$. Then also $|I_s|\geq n-|I_n|$ and, using Eq. (\ref{eq:d}), we lowerbound D(t):
\begin{equation}
D(t)\geq \frac{2z}{n} \Phi(P)(n-|I|)-r|I|.
\label{eq:gend}
\end{equation}
Bounding the latter expression in Eq. (\ref{eq:gend}) from zero we find that $D(t)$ is positive for $\frac{1}{2}\leq i\leq\frac{1}{1+\frac{rn}{2z\Phi(P)}}$, therefore the stationary fraction $i_s$ must be higher than this:
\begin{equation}
i_s\geq\frac{1}{1+\frac{rn}{2z\Phi(P)}}.
\end{equation}
Let us now focus on the opposite case, mainly $i_s\leq\frac{1}{2}$, $|I_s|\leq n-|I_n|$.
We again lowerbound D(t) using Eq. (\ref{eq:d}):
\begin{eqnarray}
D(t)\geq \frac{2z}{n} \Phi(P)|I|-r|I|\geq\frac{2z}{n} \Phi(P)|I|-r(n-|I|).
\label{eq:dfi}
\end{eqnarray}
Bounding right hand side of Eq. (\ref{eq:dfi}) from zero, we conclude analogically to the situation above:
\begin{equation}
i_s\leq\frac{1}{1+\frac{2z\Phi(P)}{rn}}.
\end{equation}
This result, however mathematically correct, appears to be quite useless: the value of $\Phi(P)$ is usually much lower than the sums that it approximates ($\sum_{k\in I(t),\,j\notin I(t)}P_{jk}$, $\sum_{k\in I(t),\,j\notin I(t)}P_{kj}$) during the process. Let us thus work out exact results for some special cases.

\subsection{Special cases}
\subsubsection{Complete graph}
For complete graphs, i.e. graphs with all possible links present, we easily find the exact solution of stationary state problem. Note, that for this special case:
\begin{equation}
\sum\limits_{k\in I(t),\,j\notin I(t)}P_{jk}=\sum\limits_{k\in I(t),\,j\notin I(t)}P_{jk}=\frac{|I(t)|(n-|I(t)|)}{n-1},
\end{equation}
as each of $|I|$ infected nodes is linked to each of $(n-|I|)$ susceptible nodes by an edge chosen with probability $\frac{1}{n-1}$ as each node has $(n-1)$ neighbors. We can thus find explicit and exact condition for $D(t)=0$. From Eq. (\ref{eq:d}) we get:
\begin{equation}
i(t)_{s}=1-\frac{r(n-1)}{2z}.  
\label{eq:fg-stat}
\end{equation}

\subsubsection{Uncorrelated homogenous graph}
\label{largen}
Let us consider now hypothetical uncorrelated homogenous graph. The term ''uncorrelated'' stands for the feature that the probability that an edge departing from a vertex of degree $j$ points on a vertex of degree $k$ is independent from the degree of vertex $j$. By ''homogenous'' we mean that average number of connections between sets of vertices of some fixed sizes depends only on these sizes, not on the actual constituents of those sets.
\newline\indent
Bearing these assumptions in mind let us compute expectation values of the two sums from Eq. (\ref{eq:d}):
\begin{eqnarray} \nonumber
 & \mathbb{E}(\sum_{j\in I(t),\,l\notin I(t)}P_{lj})=\frac{\mathbb{E}(k)}{n-1}\sum_{j\in I(t),\,l\notin I(t)}\mathbb{E}(\frac{1}{k}|k\geq 1) & \\ 
& = \frac{\mathbb{E}(k)}{n-1}\mathbb{E}(\frac{1}{k}|k\geq 1)|I(t)|(n-|I(t)|),  &
\end{eqnarray}
where we put $\mathbb{E}(k)/(n-1)$ for the expectation value of existence of link between two vertices. We substract 1 from $n$ as a node cannot be connected with itself. The stationary infected nodes density comes to be:
\begin{equation}
 i_s=1-\frac{r(n-1)}{2z\langle \frac{1}{k} \rangle\langle k \rangle},
\label{eq:gnp-stat}
\end{equation}
where we denote $\langle k \rangle=\mathbb{E}(k)$ and $\langle 1/k \rangle=\mathbb{E}(1/k)$. Specifically, for $G(n,p)$ random graph (with the well-known binomial degree distribution) the product of $\langle \frac{1}{k} \rangle\langle k \rangle$ goes to 1. In this case the latter result (\ref{eq:gnp-stat}) recovers the solution for complete graphs (\ref{eq:fg-stat}). Moreover, $G(n,p)$ graphs are indeed uncorrelated in the limit of large $n$ \cite{gnp-nocor}, so we expect $G(n,p)$ behaving like complete graphs for large $n$.

\subsection{Mixing time}
\label{ms}
In this chapter we will be interested in mixing time described in this article, i.e. the time needed by the process to reach the stationary state. Strictly speaking, this is kind of meta-stable stationary state, as in simulations on finite networks the only absorbing, stable state is the situation when the number of infected nodes is zero. It is clearly visible on Fig. \ref{fig:singlerun}, that we can distinguish two regimes with different behaviour: the regime of rapid increase in the number of infected nodes and the regime of stabilization. Let us 
state and prove a general theorem restricting mixing time for any graph. The proof is inspired by related considerations for gossip spreading done by Shah \cite{shah}.

\begin{thm}
Let $P$ be a stochastic transition matrix of a graph $G$ of the size $n$. Then the mixing time $T$ for the process described above fulfills:
$$ T(\epsilon)=O(\log n + \log \epsilon^{-1}). $$ 
\label{thm:1}
\end{thm}

\begin{proof}
We devide the proof into two parts, considering separately two stages of the process evolution:
$|I(t)|\leq\frac{n}{2}$ and $|I(t)|\geq\frac{n}{2}$. 
\begin{itemize}
  \item $|I(t)|\leq\frac{n}{2}$
\end{itemize}
We recall first the general result for $i_s\leq\frac{1}{2}$ stated in Eq. (\ref{eq:dfi}):  
$$D(t)\geq \frac{2z}{n} \Phi(P)|I(t)|-r|I(t)|.$$
Denote now by $\Lambda$ the smallest time $t$ such that the number of infected nodes exceeds $\frac{n}{2}$:
$$\Lambda=\inf\{t:|I(t)|>\frac{n}{2}\},$$
$$\Lambda \land t = \min(\Lambda,t). $$
Note, that as long as $|I(t)|\leq\frac{n}{2}$, we have $\Lambda \land (t+1)=\Lambda \land t+1$.
Recall now the general feature for any convex function $g$, $x_1,\,x_2\in\mathbb{R}$:
\begin{equation}
 g(x_1)\leq g(x_2)+g'(x_1)(x_z-x_2).
\end{equation}
Let us take: $g(x)=\frac{1}{x}$, $x_1=|I(t+1)|$ and $x_2=|I(t)|$, then:
\begin{equation}
\frac{1}{|I(t+1)|}\leq\frac{1}{|I(t)|}-\frac{1}{|I(t+1)|^2}\Big(|I(t+1)|-|I(t)|\Big).
\label{eq:appl-conv}
\end{equation}
By construction of the process we have:
$$|I(t+1)|\leq|I(t)|+1=d |I(t)|,$$
where $1\leq d\leq 2$, but as $|I(t)|=O(n)$ for $n$ big enough the constant $d$ can be arbitrarily close to $1$.
Now we continue with Eq. (\ref{eq:appl-conv}):
\begin{eqnarray}\nonumber
&\frac{1}{|I(t+1)|}\leq \frac{1}{|I(t)|} - \frac{1}{d^2|I(t)|^2}\Big(|I(t+1)|-|I(t)|\Big)\leq &\\ \label{eq:est}
&\frac{1}{|I(t)|}-\frac{1}{d^2|I(t)|^2}\Big( \frac{2z}{n} \Phi(P)|I(t)|-r|I(t)| \Big)\leq&  \\ \nonumber
&\frac{1}{|I(t)|}\Big(1- (\frac{2z}{n} \Phi(P)-r)d^{-2} \Big) \leq  \frac{1}{|I(t)|} \exp(-\frac{1}{d^2}(\frac{2z}{n} \Phi(P)-r)),& 
\end{eqnarray}
where in the second line we used Eq. (\ref{eq:dfi}) and the definition of $D(t)$, Eq. (\ref{eq:d}). In the last line we used the fact that $1-x\leq \exp(-z)$. Let us now define:
\begin{eqnarray}\label{eq:where}
&\zeta(t)=\frac{\exp(at)}{|I(t)|},& \\  \nonumber
&where\,\,a=\frac{1}{d^2}(\frac{2z}{n} \Phi(P)-r)&
\end{eqnarray}
We show that $\zeta(t)$ is a supermartingale, i.e. $ \mathbb{E}(\zeta(t)|\{\zeta(s) : s \le t' \}] \le \zeta(t') \quad \forall t' \le t$. As the only component of $\zeta(t)$ which is a random variable is $I(t)$ and as the process we analyze is Markovian and as $\Lambda \land (t+1)=\Lambda \land t+1$, it is enough to show that $\mathbb{E}(\zeta(\Lambda \land (t+1)) | I(\Lambda \land t) )\leq \zeta(\Lambda \land t)$. We do it using Eq. (\ref{eq:est}):
\begin{eqnarray}\nonumber
&\mathbb{E}(\zeta(\Lambda \land (t+1)) | I(\Lambda \land t) )=& \\ 
&\exp((\Lambda \land t) a) \exp(a)\mathbb{E}(\frac{1}{|I(\Lambda \land t+1)|}|I(\Lambda \land t))\leq&\\ \nonumber
& \exp((\Lambda \land t) a) \exp(a)\frac{1}{|I(\Lambda \land t)|} \exp(-a)= \zeta(\Lambda \land t).&
\end{eqnarray}\label{eq:gnp-pop}
As $\zeta(t)$ is a supermartingale we conclude that $\mathbb{E}(\zeta(\Lambda \land t))\leq\mathbb{E}(\zeta(\Lambda \land 0))=1$.
Furthermore, as we restrict ourselves to $|I(t)|\leq\frac{n}{2}$:
\begin{equation}
\zeta(\Lambda\land t)\geq \frac{2}{n} \exp((\Lambda\land t)a),  
\label{eq:z}
\end{equation}
and directry from it we conclude that:
\begin{equation}
\mathbb{E}(\exp((\Lambda\land t)a) )\leq \frac{n}{2}\mathbb{E}(\zeta(\Lambda\land t))\leq\frac{n}{2},  
\end{equation}
where in the last step we used the supermartingale property. Moreover, as $\exp((\Lambda\land t)a)\uparrow
\exp(\Lambda a)$ as $t\to \infty$, we have also:
\begin{equation}
\mathbb{E}( \exp(\Lambda a) )\leq \frac{n}{2}.  
\end{equation}
Finally, let us recall the Markov inequality:
\begin{equation}
  \mathbb{P}(|X|\geq c )\leq \frac{\mathbb{E}(|X|)}{c} 
\end{equation}
and choose $t_1=\frac{1}{a}(\ln(n)-\ln(\epsilon))$. Then we straightforwardly get:
\begin{equation}
\mathbb{P}(\Lambda>t_1)=\mathbb{P}(\exp(\lambda a)>\frac{n}{\epsilon})\leq \frac{ \mathbb{E}({\exp(\lambda a)})}{\frac{n}{\epsilon}}\leq\frac{\epsilon}{2}.  
\end{equation}

\begin{itemize}
  \item $|I(t)|\geq\frac{n}{2}$
\end{itemize}
For this case we perform exactly the same procedure, but starting from Eq. (\ref{eq:gend}) instead of Eq. (\ref{eq:dfi}), which we started with in the previous case. Following the same steps as above we only change constant $a$ in Eq. (\ref{eq:where}) into 
$b=\frac{1}{d^2}(\frac{2z}{|I_s|} \Phi(P)-\frac{2z}{n} \Phi(P)-r)$, where explicitely appears the number of infected nodes at the stationary state. Second thing that has to be changed is Eq. (\ref{eq:z}) where, instead of $\frac{n}{2}$ we can put $n$. Resulting time for this stage is:
\begin{eqnarray}
&\mathbb{P}(\Lambda>t_2)=\leq \frac{ \mathbb{E}({\exp(\lambda b)})}{\frac{n}{\epsilon}}\leq\epsilon,  & \\ \nonumber
& where\,\,t_2=\frac{1}{b}(\ln(n)-\ln(\epsilon)).&
\end{eqnarray}

\end{proof}

From this general theorem we conclude, that the closer we are with chosen parameters to the zero--stationary state (i.e. the smaller is the stationary density of infected nodes), the slower is the first phase of rapid increase:

\begin{cor}
Mixing time is linear with inverse of the distance $\eta$ from the epidemics treshold, i.e.:
$$ T(\eta,\epsilon)=O\Big(\frac{1}{\eta}(\log n + \log \epsilon^{-1})\Big). $$
\label{cor:1}
\end{cor}

\begin{proof}
Recall Eq. (\ref{eq:d}): we demand $D(t)\geq0$ and transform this condition to:
\begin{equation}
  \frac{z}{nr}\geq\frac{|I(t)|}{\Big(\sum\limits_{k\in I(t),\,j\notin I(t)}P_{jk}+\sum\limits_{k\in I(t),\,j\notin I(t)}P_{kj}\Big)},
\end{equation}
which boils down to equality for stationary state. We denote right hand side of this equation by $p_c$ for the smallest possible situation, i.e. for epidemics treshold. Now let us take values of parameters $z$, $n$ and $r$ such that:
\begin{equation}
  \frac{z}{nr}=p_c(1+\eta),
\label{eq:para}
\end{equation}
where $\eta\geq0$. Now we recall some parts of the proof of Theorem \ref{thm:1}. Actually, all we have to do is to rewrite condition for $D(t)$ in parametrization given in Eq. (\ref{eq:para}) and notion of $p_c$:
\begin{equation}
D(t)\geq r\Big(\frac{z}{nr} \frac{1}{p_c} |I(t)| -|I(t)|\Big)=r\eta |I(t)|.
\end{equation}
We put this result into Eq. (\ref{eq:est}) obtaining:
\begin{equation}
\frac{1}{|I(t+1)|}\leq \frac{1}{|I(t)|} \exp(-\frac{r\eta}{d^2}),
\end{equation}
end then we proceed in the same way as in the proof of Theorem \ref{thm:1}. The result is
\begin{eqnarray}
&\mathbb{P}(\Lambda>t_c)=\leq \frac{ \mathbb{E}({\exp(\lambda \frac{r\eta}{d^2} )})}{\frac{n}{\epsilon}}\leq\epsilon,  & \\ \nonumber
& where\,\,t_c=\frac{d^2}{r\eta}(\ln(n)-\ln(\epsilon)).&
\end{eqnarray}
   
\end{proof}

\section{Simulation}
Here we present simulations for stationary state of various types of networks, i.e. complete graph, $G(n,p)$ random graph \cite{Gilbert}, Watts-Strogatz small world graph \cite{SW} and graphs with power law degree distribution (scale-free network, see e.g. \cite{NewmanSIAM}). Computer-simulational investigations focus on the topics described theoretically in the last section, i.e. epidemic treshold, stationary state and mixing time.

\subsection{Epidemic treshold}
We check here the behaviour of the process in the very beginning, i.e. exactly at the first time step. Four kinds of networks are being examined: complete graph, $G(n,p)$ random graph with $p=0.5$, small world graph with $k=6$ neighbors on the circle and rewiring probability $p=0.5$ (see \cite{SW}) and scale-free network with the exponent $\alpha=2.5$. We vary sizes of networks $n$ and for each type of the graph we choose different recovery probability $r$. Looking for the critical value of contamination probability $z_c$ we change parameter $z$ and check for which value the fraction of infected nodes starts to increase. This procedure is being repeated 100000 times. Results are presented in Fig. \ref{fig:1step}. Visibly, simulations follow the theoretical prediction Eq. (\ref{eq:trsh}) prefectly for all four kinds of graphs being examined.

\begin{center}
\begin{figure}
  \includegraphics[height=5cm]{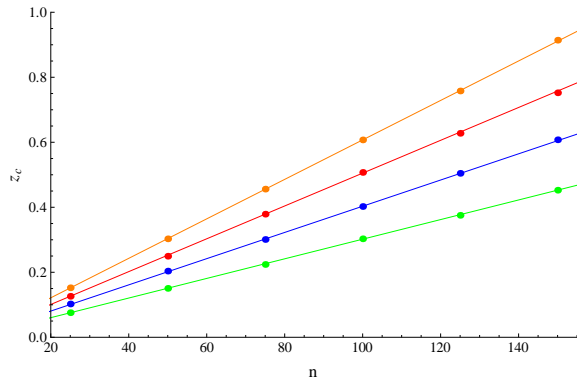}
  \caption{Epidemic treshold for four different type of graphs: dots stay for simulational results, lines present theoretical prediction, Eq. (\ref{eq:trsh}). Starting from the bottom we have results for $G(n,p)$ random graph (green line, $r=0.006$), scale-free network (blue line, $r=0.008$), complete graph (red line, $r=0.01$) and small world graph (orange line, $r=0.012$)}
\label{fig:1step}
\end{figure}
\end{center}

\subsection{Stationary state}
 Results for stationary state are obtained by performing many runs (typically 1000), finishing each of them at a fixed, long time step (10 000 - 100 000), cutting the beginning phase of rapid increase and fitting a line to the points oscillating about the stationary state. There are two types of results which we can end up with after a single run: epidemics either dies at a certain point (i.e. number of infected nodes, due to fluctuations, reaches zero and - by construction of the model - stays zero, usually it happens at the very beginning of the process) or number of infected nodes increases rapidely in the first stage, and then oscillates over some fixed value (see Fig.\ref{fig:singlerun}). We call this value stationary state (presicely, as we have already noted in Sec. \ref{ms}, meta-stable state). In order to compute average stationary state we neglect all the runs where there exist such a time step, when the number of infected nodes equals zero. 

\begin{center}
\begin{figure}
  \includegraphics[height=5cm]{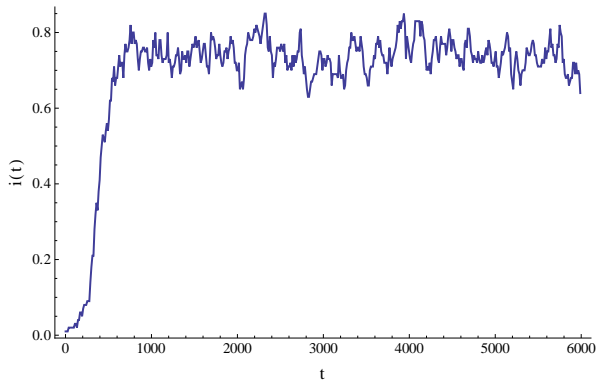}
  \caption{An example of a single run for $z$=1, $r$=0.005, random graph $G(n,p)$ of the size $n$=100 and $p$=0.5.}
\label{fig:singlerun}
\end{figure}
\end{center}
\begin{center}
\begin{figure}
  \includegraphics[height=5cm]{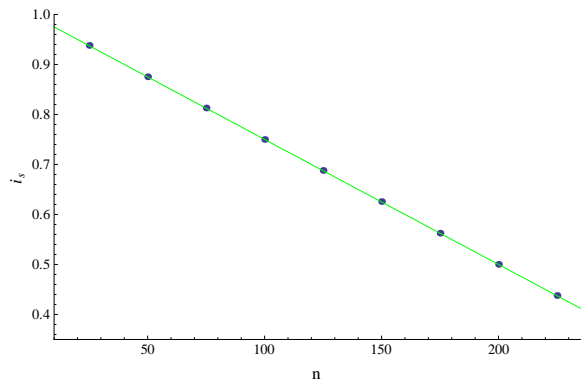}
  \caption{Plot of stationary state value of infected nodes denstiy $i_s$ for complete graphs versus network size $n$: simulation (blue dots) and theoretical result (\ref{eq:fg-stat}) (green line). We fix here $z$=1, $r$=0.005.}
\label{fig:fg-stat-n}
\end{figure}
\end{center}
\begin{center}
\begin{figure}
  \includegraphics[height=5cm]{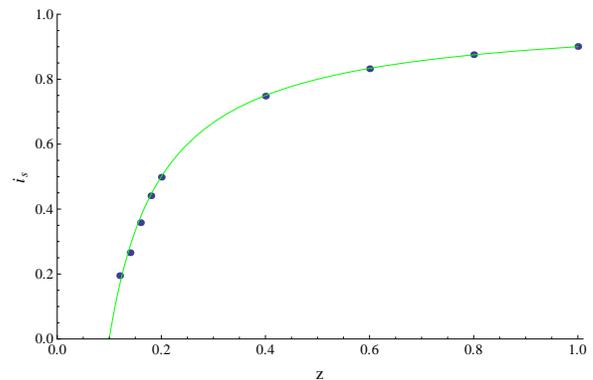}
  \caption{Plot of stationary state value of infected nodes denstiy $i_s$ for complete graphs versus contamination probability $z$: simulation (blue dots) and theoretical result (\ref{eq:fg-stat}) (green line). We fix here $n$=100, $r$=0.005.}
  \label{fig:fg-stat-z}
\end{figure}
\end{center}
First we examine complete graphs, as in the last section we provided the exact result for them (\ref{eq:fg-stat}). In Fig. \ref{fig:fg-stat-n} we show how stationary infected nodes density $i_s$ depends on network size $n$. Then, in Fig.\ref{fig:fg-stat-n}, we show dependence on contamination probability $z$. Both figures show perfect agreement between simulation and theory, Eq. (\ref{eq:fg-stat}).

As we have already seen the behaviour of complete graphs and how they relate to the theory described above, let us compare stationary state $i_s$ for four different kinds of graphs. In Fig. \ref{fig:ogolne} we show the results for complete graph, $G(n,p)$ random graph with $p=0.1$, small world graph with $k=10$ neighbors on the circle and rewiring probability $p=0.5$ (see \cite{SW}) and scale-free network with the exponent $\alpha=2.5$. Sizes of the graphs are fixed, $n=100$. Noticeably, the results for three out of four kinds of graphs are almost the same, while scale-free network goes an entirely different way. Below we will focus on complete, $G(n,p)$ and small world graphs only.

\begin{center}
\begin{figure}
  \includegraphics[height=5cm]{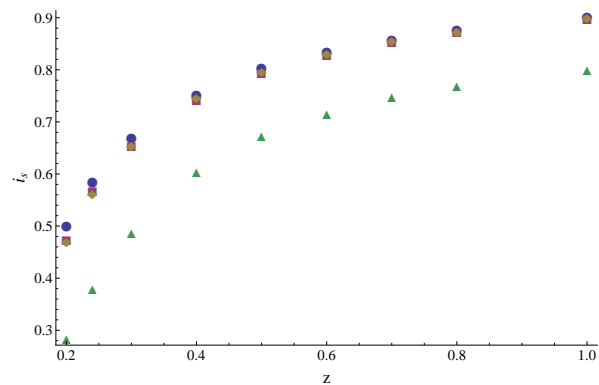}
  \caption{Plot of stationary state value of infected nodes denstiy $i_s$ for complete graph (blue circles), $G(n,p)$ random graph with $p=0.1$ (red squares), small world graph with $k=10$ neighbors on the circle and rewiring probability $p=0.5$ (yellow rotated squares) and scale-free network with the exponent $\alpha=2.5$ (green triangles) versus contamination probability $z$ and fixed network size $n=100$.}
\label{fig:ogolne}
\end{figure}
\end{center}

In Sec. \ref{largen} we concluded, that $G(n,p)$ graphs for large $n$ should resemble like complete graphs. It is instructive to see that in the limit of large $n$, epidemics, not only on $G(n,p)$, but also on small world graphs behaves the same as on complete graphs, see Fig. \ref{fig:srednieodn}.

\begin{center}
\begin{figure}
  \includegraphics[height=5cm]{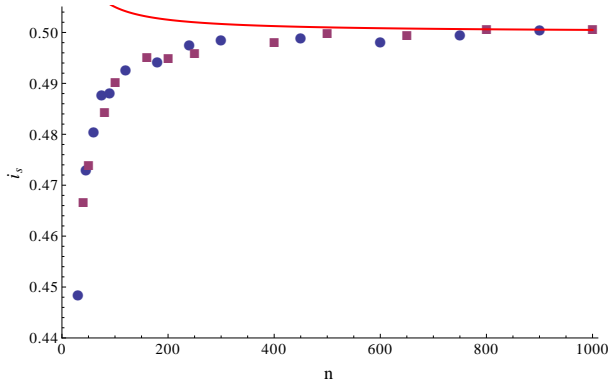}
  \caption{Plot of stationary state value of infected nodes denstiy $i_s$ for $G(n,p)$ random graph with $p=0.2$ (blue circles) and small world graph with rewiring probability $p=0.5$ (red squares) versus network size $n$. Number of neighbors on the circle $k=2n/10$ is chosen such that the edges density $\frac{k}{2n}$ stays fixed. Red line shows theoretical prediction for complete graphs (\ref{eq:fg-stat}). We fix here $z=1$ and $n\times r=1$.}
\label{fig:srednieodn}
\end{figure}
\end{center}

\subsection{Mixing time}
In this section we examine mixing times of the process, i.e. we check how long does it take to reach stationary state. Fig. \ref{fig:czasy_n} depicts how does average mixing time depend on $ln(n)$, where $n$ is network size, as usually. This is done for complete graph, $G(n,p)$ random graph with $p=0.2$ and small world graph with rewiring probability $p=0.5$. For the same graphs we check average mixing time dependence on inverse of distance from epidemics treshold $\eta$ (see Corrolary \ref{cor:1}). It is shown in Fig. \ref{fig:czasy_e}.
\begin{center}
\begin{figure}
  \includegraphics[height=5cm]{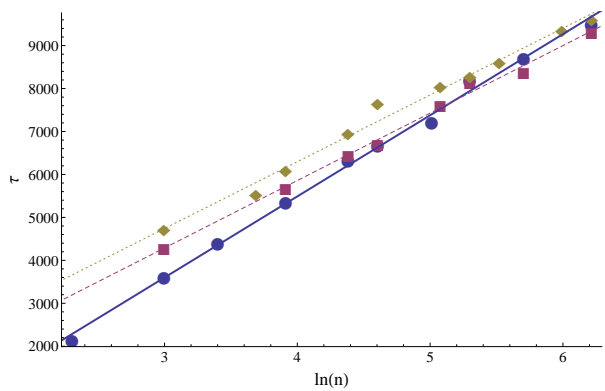}
  \caption{Average mixing time $\tau$ for complete graph (blue dots), $G(n,p)$ random graph with $p=0.2$ (red squares), small world graph with rewiring probability $p=0.5$ (yellow rotated squares) versus logarithm of network size $ln(n)$. Number of neighbors on the circle $k=2n/10$ is chosen such that the edges density $\frac{k}{2n}$ stays fixed. Simulational results are depicted by blue dots and red line shows theoretical prediction for complete graphs (\ref{eq:fg-stat}). We fix here $r=0.001$ and $n/z=1000$ in order to have stationary state not changed. Lines are plotted to guide the eye.}
\label{fig:czasy_n}
\end{figure}
\end{center}

\begin{center}
\begin{figure}
  \includegraphics[height=5cm]{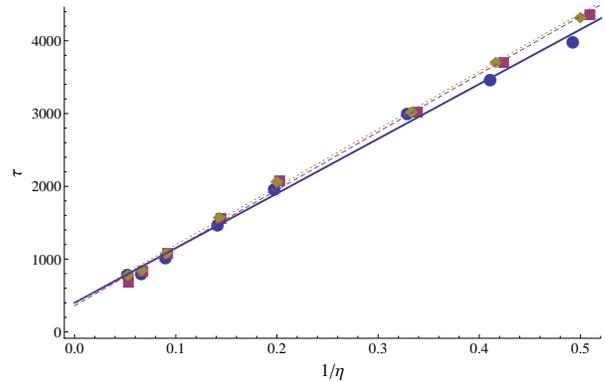}
  \caption{Average mixing time $\tau$ for complete graph (blue dots), $G(n,p)$ random graph with $p=0.2$ (red squares), small world graph with rewiring probability $p=0.5$ and $k=20$ neighbors on the circle (yellow rotated squares) versus inverse of distance from epidemics treshold $\eta$. Simulational results are depicted by blue dots and red line shows theoretical prediction for complete graphs (\ref{eq:fg-stat}). We fix here $r=0.001$ and $n=100$. Lines are plotted to guide the eye.}
\label{fig:czasy_e}
\end{figure}
\end{center}

These result show actually much more than Theorem and Corollary from Sec. \ref{ms}. We examine here average mixing time and show, that they are linear with $ln(n)$ and $1/\eta$, as theory in Sec. \ref{ms} suggest by bounds of probability of mixing time proportional to $ln(n)$ and $1/\eta$.

\section{Conclusions}
We have proposed model of epidemics spreading with at most one infection per times step. Starting from the general formula for the change of the number of infected nodes (\ref{eq:d}) we provided condition for epidemics treshold for any kind of graph. Simulational results for epidemics treshold follow the theoretical predictions perfectly. Further more, stationary density of infected nodes for complete and uncorrelated homogenous graphs has been derived and bounds for this density, using the notion of graph conductance, have been obtained. Complete graph simulations show agreement with the theory. Epidemy on $G(n,p)$ random graphs, according to no correlation in large $n$ limit \cite{gnp-nocor}, as well as on small world graphs, in the large $n$ limit, behave like epidemy on complete graphs.
\newline\indent
We have stated and proven theorem and corollary that bouds the probability of mixing time by values proportional to $ln(n)$ and $1/\eta$, where $n$ and $\eta$ are size of the network and distance form  epidemics treshold respectively. Simulations on complete, $G(n,p)$ and small world graphs show even more, mainly that the average mixing time is linear with $ln(n)$ and $1/\eta$.

\section*{Acknowledgements}
It is a pleasure to thank D.~Kwietniak and P.~De Los Rios
for fruitful discussions and helpfull advices.


\begin{thebibliography}{11}
\bibitem{AM} R.M.Anderson, and R.M.May, {\sl Infectious Diseases of Humans. Dynamics and Control} (Oxford University Press, 1992).
\bibitem{NewmanSIAM} M.E.J.Newman, {\sl SIAM Rev.} {\bf 45}, 167 (2003).
\bibitem{Boguna} A.S.Saumell-Mendiola, M.\'Angeles Serrano, and M.Bogu\~n\'a, {\sl Phys. Rev. E} {\bf 86}, 026106 (2012).
\bibitem{DynProc} A.Barrat, M.Barthélemy, A.Vespignani, {\sl Dynamical Processes on Complex Networks} (Cambridge University Press, 2008).
\bibitem{PSV} R.Pastor-Satorras, and A.Vespignani, {\sl Phys. Rev. Lett.} {\bf 86}, 3200 (2001).
\bibitem{cor1} R.Pastor-Satorras, and A.Vespignani, {\sl Evolution and Structure of the Internet: A Statistical Physics Approach} (Cambridge University Press, 2004).
\bibitem{BPSV} M.Bogu\~n\'a, R.Pastor-Satorras, and A.Vespignani, {\sl Phys. Rev. Lett.} {\bf 90}, 028701 (2003).
\bibitem{cor2} M.Bogu\~n\'a, and R.Pastor-Satorras, {\sl Phys. Rev. E} {\bf 68}, 036112 (2003).
\bibitem{cor3} Y.Moreno, J.B.GG\'omez, and A.F.Pacheco, {\sl Phys. Rev. E} {\bf 68}, 035103 (2003).
\bibitem{FullM} S.G\'omez, A.Arenas, J.Borge-Holthoefer, S.Meloni, and Y.Moreno, {\sl Europhys. Lett.} {\bf 89}, 38009 (2010).
\bibitem{Localization} A.V.Goltsev, S.N.Dorogovtsev, J.G.Oliveira, and J.F.F.Mendes, {\sl Phys. Rev. Lett.} {\bf 109}, 128702 (2012).
\bibitem{Paulo} T.Petermann, P.De Los Rios, {\sl J. Theor. Biol.} {\bf 229}, 1 (2004).
\bibitem{gnp-nocor} M.E.J.Newman, {\sl Phys. Rev. Lett.} {\bf 89}, 208701 (2002). 
\bibitem{shah} D.Shah, {\sl Foundations and Trends \textregistered in Networking} {\bf 3}, 1 (2009).
\bibitem{Gilbert} E.N.Gilbert, Random graphs, {\sl Ann. Math. Stat.} \textbf{30}, 1141 (1959).
\bibitem{Con} A.Sinclair, {\sl Algorithms for random generation and counting: a Markov chain approach} (Birkhauser Verlag, 1993)
\bibitem{SW} D.J.Watts, S.H.Strogatz, {\sl Nature} \textbf{393}, 440 (1998).
\end{thebibliography}
\end{document}